\documentclass[pdf]{article}

\newcommand{\nimrodbug}[1]{#1}

\usepackage{multirow}
\usepackage{url}
\usepackage{subcaption}
\usepackage{graphicx}
\usepackage{color}
\usepackage{epsfig}
\usepackage{enumerate}
\usepackage{wrapfig}
\usepackage{nicefrac}
\usepackage{pgf}
\usepackage{booktabs}

\usepackage{amsmath}
\usepackage{amsfonts}
\usepackage{amsthm}
\usepackage{amssymb}
\usepackage{todonotes}

\newcommand{\calS}{\mathcal{S}}

\newcommand{\naturals}{{{\mathbb{N}}}}
\newcommand{\maj}{{\mathrm{maj}}}
\newcommand{\unit}{{\mathrm{unit}}}
\newcommand{\full}{{\mathrm{full}}}

\newcommand{\p}{{\mathrm{P}}}
\newcommand{\np}{{\mathrm{NP}}}
\newcommand{\fpt}{{\mathrm{FPT}}}

\newtheorem{theorem}{Theorem}

\newtheorem{corollary}[theorem]{Corollary}

\newtheorem{remark}{Remark}
\newtheorem{definition}{Definition}

\newtheorem{experiment}{Experiment}
\newtheorem{proposition}[theorem]{Proposition}

\newtheorem{observation}{Observation}
\newtheorem{question}{Question}

\newcommand{\calR}{\mathcal{R}}

\def\row#1#2{{#1}_1,\ldots ,{#1}_{#2}}

\def\2vec#1#2{\left(\begin{array}{c}{#1}\\{#2}\end{array}\right)}

\newcommand{\approval}{{\mathrm{s}}}

\newcommand{\mypara}[1]{\smallskip\noindent\textsc{#1.}}

\title{The Complexity of Multiwinner Voting Rules with Variable Number of Winners}

\author{
  \makebox[0.35\linewidth]{Piotr Faliszewski} \\ 
  {AGH University} \\ 
  {Krakow, Poland}
\and 
  \makebox[0.35\linewidth]{Arkadii Slinko} \\
  {University of Auckland} \\
  {Auckland, New Zealand}
\and 
  \makebox[0.35\linewidth]{Nimrod Talmon} \\
  {Weizmann Institute of Science} \\
  {Rehovot, Israel}
}

\newcommand{\shortcite}[1]{\cite{#1}}

\begin{document}

\maketitle

\begin{abstract}
  We consider the approval-based model of elections, and undertake a
  computational study of voting rules which select committees whose
  size is not predetermined.  While voting rules that output committees with a
  predetermined number of winning candidates are quite well studied,
  the study of elections with variable number of winners %
  has only recently been initiated by Kilgour~\cite{kil:j:variable-size-committee}.
  This paper aims at achieving a better understanding of these rules,
  their computational complexity, and on scenarios for which they might be applicable.

\end{abstract}

\section{Introduction}

We study the setting where a group of agents (the voters) want to
select a set of candidates (a committee) based on these agents'  preferences.
Agents are asked which candidates they approve of for the inclusion into the committee and this input data
needs to be
aggregated. %
However, as opposed to the quickly growing body of work on 
electing committees of a fixed
size~\cite{elk-fal-sko-sli:j:multiwinner-properties,fal-sko-sli-tal:c:hierarchy-committee,azi-gas-gud-mac-mat-wal:c:approval-multiwinner,ama-bar-lan-mar-rie:c:minisum-to-minimax,kil:chapter:approval-multiwinner},
here we are interested in rules that derive both the size of the
winning committee and its members from the voters' preferences.
Recently, Kilgour~\cite{kil:j:variable-size-committee} %
and Duddy et al.~\cite{Duddy2016} initiated a systematic study of
such voting rules;
here we are  interested in the complexity of
computing their winners and in 
experimentally analyzing the sizes of the elected
committees
(for some early axiomatic results, we also
point the reader to the work of Brandl and
Peters~\cite{bra-pet:i:borda-mean}). 

\subsection{When Not To Fix the  Size of the Committee?}
There is a number of settings where it is not natural to fix the size
of the committee to be elected and it is better to deduce it from the
votes.  Since so far committee elections with variable number of
winners did not receive much attention in the AI literature, below we
provide a number of examples of such settings.
We do not mention this repeatedly, but one may wish to automate the
processes in the examples below using AI techniques.

\smallskip

\noindent\textbf{Initial Screening.}\quad
Consider a situation where we need to select one item---among many
possible ones---that has some desirable features.  The final decision
can only be done by a qualified expert, but we have a number of easy
to evaluate (but imperfect) criteria that the selected items should
satisfy (these criteria are soft and it may be that the best item
actually fails some of them). We view each criterion as a voter (who
``approves'' the items that satisfy it) and we seek a committee,
hopefully of a small size, of candidates from which the qualified
expert will choose the final item.\footnote{One of the authors of this
  paper was once tasked with the problem of classifying a collection
  of daggers for a museum. The solution was to compute a number of
  partitions of the set of daggers into clusters, evaluate their
  qualities without reference to ethnographical knowledge on the daggers,
  and to present the best ones to the museum's experts, who chose one
partition that led to an ethnographically meaningful classification.} %

Initial screening is closely related to
shortlisting~\cite{bar-coe:j:non-controversial-k-names,elk-fal-sko-sli:j:multiwinner-properties}. We
use a different name for it to emphasize that we do not fix the number
of candidates to choose, as is the case with shortlisting.\smallskip

\noindent\textbf{Finding a Set of Qualifying Candidates.}\quad 
Finding a set of candidates that satisfy all or almost all criteria is
a common problem.  Real-life examples include selecting baseball
players for inclusion into the Hall of Fame and selecting students to receive an honors degree. %
In the former case, eligible voters (baseball writers) approve up to
ten players and those approved by at least 75\% of the voters are
chosen to the Hall of Fame. In the latter case, the voting process is
typically implicit; the university announces a set of criteria of
excellency---which act as voters, ``approving'' the students that
satisfy them\footnote{The criteria may include, e.g., never receiving
  a low grade from a course, taking some advanced classes, never being
  suspended, etc.}---and set rules such as ``a student receives an honors
degree if he or she meets at least five out of six criteria''. It is
often desirable that the selected committee is %
small (say, at most a few people for the Hall of Fame and some
not-too-large percentage of the students for the honors degree), but
this is not always the case. E.g, consider the task of
selecting people for an in-depth medical check based on a number
of simple criteria that jointly indicate elevated risk of a certain
disease; everyone who is at risk should be checked regardless of the
number of those patients.

One of the first procedures formally proposed for the task of
selecting a group of qualifying candidates was the majority rule (MV), suggested by
Brams, Kilgour, and Sanver~\cite{brams2007minimax}.  The majority
rule outputs the committee that includes all the candidates that are
approved by at least half of the voters (satisfy at least half of the criteria).
It is, of course, natural to consider MV with other thresholds, as, for example, in the Hall of Fame
example.
\smallskip

\noindent\textbf{Partitioning into Homogeneous Groups.}
For the case of partitioning candidates into \emph{homogeneous}
groups, we can no longer focus only on one of the groups (the
committee),
but rather we care about a partition into two groups so that each of
them contains candidates that are as similar as possible. A prime
example here is partitioning students in some class into two groups,
e.g., a group of beginners and a group of advanced ones (say,
regarding, their knowledge of a foreign language; depending on the
setting, it may or may not be important to keep the sizes of these two
groups close). The students are partitioned in this way to facilitate
a better learning environment for everyone; in the context of voting,
the issue of partitioning students was raised by Duddy et
al.~\cite{Duddy2016}. \smallskip

\iffalse
\noindent\textbf{Choosing Committees with Implicit Proferences on Size.}
Consider the problem of choosing a group of specialists to hire.  This
problem is similar to that of partitioning the candidates into two
non-homogeneous groups, where the first group constitutes the team of
specialists to be hired (e.g., a group of software developers for a
project, evaluated on their knowledge of PHP, Java, previous
experience, etc.), but we additionally have implicit preferences
regarding the committee size (e.g., we may wish to hire at least five
developers, but from prior experience we know that a team of more than
ten people is hard to manage).

For another example in this group, consider a factory that can produce
a number of different items. The factory should produce those that are
most popular on the market (that are ``approved'' by many possible
customers), but it should not produce too many different kinds as it
may involve increased costs of various kinds (e.g., adding a new
product may require expensive extension of the factory, which would
only pay off if the new product were very popular).  \smallskip
\fi

\noindent\textbf{Finding a Representative Committee.} %
An elected committee is representative if each voter approves at least
one committee member (who then can represent this voter).  The idea of
choosing a representative committee of a fixed size received
significant attention in the literature (see the works of Chamberlin
and Courant~\cite{cha-cou:j:cc}, Monroe~\cite{mon:j:monroe},
Elkind et al.~\cite{elk-fal-sko-sli:j:multiwinner-properties}, and
Aziz et
al.~\cite{azi-bri-con-elk-fre-wal:j:justified-representation} as
some examples). However, as pointed out by Brams and
Kilgour~\cite{brams2014satisfaction}, committees of fixed size
simply cannot always provide adequate representation. Thus, in some
applications, it is natural to elect committees without
prespecified sizes (while in others, fixing the committee size may be
necessary).

A representative committee may be desired when some authorities are
revising existing regulations 
and need to consult citizens, for which purpose they would like
to select a representative focus groups in various cities. The people
in a focus group do not have to represent the society proportionally
(their role is to voice opinions and concerns and not to make final
decisions), but should cover all the spectrum of opinions in the society. 
Usually, small representative committees are more desirable than larger ones. 

\subsection{Our Contribution}

Our main goal
is to study computational aspects of
voting rules tailored for elections with a variable number of
winners. This direction was pioneered by Fishburn and
Peke\v{c}~\cite{fishburn2004approval}, who introduced the class
of threshold rules and studied their computational complexity
(somewhat surprisingly, only for the case when the committee size is
fixed).  We are not aware of other computational studies that
followed their work. %

In this paper we study the threshold rules of Fishburn and
Peke\v{c}~\cite{fishburn2004approval}, as well as a number of
other rules, including those discussed by
Kilgour~\cite{kil:j:variable-size-committee}.  We obtain the
following results:
\begin{enumerate}
\item For each of the rules, we establish whether finding a winning
  committee under this rule is in $\p$ or is $\np$-hard (in which case
  we seek $\fpt$ algorithms parameterized by the numbers of candidates
  and by the number of voters).
\item We evaluate experimentally the average sizes of committees
  elected by our rules. %
  We consider a basic model of preferences, where each voter approves
  each candidate independently, with some probability $p$ (we focus on
  $p=\nicefrac{1}{2}$ but for some representative rules we consider a
  larger spectrum of probability values).
\end{enumerate}

We only make %
preliminary comments regarding suitability of our rules for the tasks
outlined above.  While specifying the applications is important to
facilitate future research, we believe that we are still at the level
of identifying voting rules and gathering basic knowledge about them.

\section{Preliminaries}
An approval-based election $E = (C,V)$ consists of a set $C = \{c_1,
\ldots, c_m\}$ of candidates and  a collection $V = (v_1, \ldots,
v_n)$ of voters.  Voters express their preferences by
filling approval ballots. The approval ballot of a voter specifies the subset of candidates that
this voter approves. To simplify notation, we denote voter $v_i$'s
approval ballot also as $v_i$ (whether we mean the voter or the subset of approved 
candidates  will always be clear from the
context).  A collection $V$ of voters, interpreted as a collection of
approval ballots in a certain election, is called the {\em preference profile}. For a given subset
$S$ of the candidates from set $C$, by $\overline{S}$ we mean the
candidates not in $S$, i.e., $\overline{S} = C \setminus S$. By the {\em approval score} of a candidate in an election, we mean the
number of voters that approve of this candidate.

A voting rule for elections with a variable number of winners is a
function $\calR$ that, given an election $E = (C,V)$, returns a family
of subsets of $C$ (the set of committees which tie as winners).  The
main point of difference between the type of voting rules that we
study here and the voting rules typically studied in the context of
multiwinner elections is that we do not fix the size of the committee
to be elected and we let it %
be deduced by the rule. %

For an overview of multiwinner election procedures using approval
balloting, we point to the works of
Kilgour~\cite{kil:chapter:approval-multiwinner,kil:j:variable-size-committee},
both for the discussions of rules with fixed and variable number
of winners; Duddy et al.~\cite{bra-pet:i:borda-mean} and Brandl
and Peters~\cite{Duddy2016} discuss the polynomial-time computable Borda mean rule (not
included in our discussion).

Our hardness results follow by reductions from the $\np$-complete problem \textsc{Set Cover}.
An instance of \textsc{Set Cover}
consists of a set $U = \{u_1, \ldots, u_n\}$ of elements, a family
$\calS = \{S_1, \ldots, S_m\}$ of subsets of $U$, and an integer
$k$; we ask whether exist at most $k$ sets from $\calS$ whose union is~$U$.

\section{Simple Approval Rule} 

We start our discussion by considering the Approval rule (AV), one of
the arguably simple rules for the setting with variable number of winners.

\begin{description}
\item[Approval Voting (AV).]  Under AV we output the (unique)
  committee of candidates with the highest approval score.
\end{description}
In essence, AV is the single-winner Approval rule which instead of
breaking ties (among winning candidates) outputs all the candidates with the highest score. The
rule is, of course, polynomial-time computable.

\begin{proposition}
  There is a polynomial-time algorithm that computes the unique
  winning committee under the AV rule.
\end{proposition}

We should expect the winning committees for this rule to be very small and, indeed,
the following experiment confirms this intuition (the AV rule is so
simple that the experiment is not really necessary; we include it for the
sake of completeness and to  provide
the setup for experiments regarding less intuitive rules). %

\begin{experiment}\label{exp:main}
  We consider elections with $m=20$ candidates and $n=20$ voters,
  where %
  for each candidate $c$, each voter approves $c$ with probability $p =
  \nicefrac{1}{2}$. We have generated $10,000$ elections and the
  average committee size was $1.52$ with standard deviation $0.89$.
  (See Table~\ref{tab:collection} for the list of average committee
  sizes in this setting for our rules; to see how the number of voters
  affects the average committee sizes, we also included
  experiments for $20$ candidates and $100$ voters).
  We repeat this experiment for all the rules in this paper.
\end{experiment}

\begin{remark}
  We have chosen a fairly small number of candidates and voters for
  our experiments as the results for such elections are fast to
  compute (even for $\np$-hard voting rules) and, yet, they appear to
  be sufficient to show the effects that we are interested in. We
  chose the model where each candidate is approved or disapproved
  independently by each voter because it is the most basic scenario
  which, we believe, one should start with (other scenarios should be
  considered in future works).
\end{remark}

\section{(Generalized) Net-Approval Voting}\label{sec:gnav}

In the framework where the size of the target committee is fixed, one
may ask for a committee of candidates whose sum of approval scores is
the highest. To adapt this idea to the variable number of winners, Brams and Kilgour~\cite{brams2014satisfaction} suggested %
the Net-Approval Voting (NAV) rule. This rule pays attention not only to approvals but also to disapprovals.

\begin{description}
\item[Net Approval Voting (NAV).]  The score of a committee $S$ in
  election $E = (C,V)$ under NAV is defined to be $ \sum_{v_i \in
    V}\big(|S \cap v_i| - |S \cap \overline{v_i}|\big)$; the
  committees with the highest score tie as co-winners.
\end{description}

Note that this rule is %
very close to the MV rule~\cite{brams2007minimax} mentioned in the
introduction; the winning committees under NAV consist of all
candidates approved by a strict majority of the voters and any subset
of those approved by exactly half of the voters (MV includes all 
candidates approved by at least half of the voters).

\begin{corollary}
  There is a polynomial-time algorithm that computes the unique smallest\footnote{%
  In fact, there is a polynomial-time algorithm that computes a winning committee of any size, if exists.}
  winning committee under the NAV rule.
\end{corollary}

\begin{experiment}
  Repeating Experiment~\ref{exp:main} for NAV, we obtain $8.25$ as the average
  size of the smallest elected committee with standard
  deviation $2.19$. This confirms the intuition that slightly fewer
  than half of the candidates would be elected in a typical election, where
  each candidate is approved independently with probability $\nicefrac{1}{2}$ 
  (since in the smallest winning committee it is necessary to be
  approved by \emph{more than} half of the voters). 
  We also computed the average size of a NAV committee for the case
  where voters approve candidates with other
  probabilities. Specifically, for each $p \in \{0.05, 0.1, \ldots,
  0.95\}$ we generated $10,000$ elections with $20$ candidates and
  $20$ voters, where each voter approves each candidate with
  probability $p$, and we computed the average size of the NAV
  committee. We repeated the same experiment for $20$ candidates and
  $100$ voters. The results are presented in Figure~\ref{fig:size}.  We see that when the number of voters becomes
  large, the graph becomes very close to the step function. This means
  that NAV should only be used in very specific settings (such as the
  baseball Hall of Fame example).
\end{experiment}

It turns out that, using the main idea behind the NAV rule, it is
possible to express many different voting rules.  Below we suggest a
language for describing such rules.

\begin{description}
\item[Generalized NAV.]  Let $f$ and $g$ be two non-decreasing,
  non-negative-valued functions, $f,g \colon
  \naturals \rightarrow \naturals$, such that $f(0)=g(0)=0$.  We define
  the $(f,g)$-NAV score of a committee $S$ in election $E = (C,V)$ to
  be:
  \[\textstyle
  \sum_{v_i \in V}\big(f(|S \cap v_i|) - g(|S \cap
  \overline{v_i}|)\big).
  \]
  The committees with the highest score tie as co-winners. The
  intuition for this rule is that we would like to be able to count
  approvals and disapprovals differently. %
  E.g., this can be explained as follows: at times, the lack of approval of a
  candidate is not really a disapproval but lack of information about
  him/her or simply no firm opinion.
\end{description}

\begin{remark}
  It would also be reasonable to include the terms $f'(|\overline{S}
  \cap \overline{v_i}|)$ and $-g'(|\overline{S} \cap v_i|)$ (for two
  additional functions $f'$ and $g'$) in the definition of the score
  above.  The first term, for example, would reflect the utility that
  voter $v_i$ has from exclusion of candidates whom he/she did not
  approve.
\end{remark}

$(f,g)$-NAV rules are quite diverse. For example, if $f$ and $g$ are
linear functions (e.g., $f(x) = x$ and $g(x) = 2x$) then $(f,g)$-NAV
is a variant of the MV rule with a different threshold of approval
(for the given example, a candidate would be included in the committee
if it were approved by at least a $\nicefrac{2}{3}$ fraction of the
voters; thus, we refer to this rule also as $\nicefrac{2}{3}$-NAV). Such rules
seem quite appropriate for the task of choosing a set of qualifying
candidates as, for each candidate $c$, the decision whether to
include $c$ in the committee or not is made based on approvals for
$c$ only (indeed, the decision if a patient should be sent for an
in-depth medical check should not depend on the health of other
patients).\footnote{As a side comment, we mention that such rules are
  also typical in the lobbying
  scenarios~\cite{chr-fel-ros-sli:j:lobbying,bin-erd-fer-gol-mat-rei-rot:j:probabilistic-lobbying,neh:c:lobbying-threshold}.}
Below we show that for nonlinear functions $f$ and $g$, $(f,g)$-NAV
rules might no longer have this independence property.

Let us consider the function $t_1(x)$,
where
$t_1(0)=0$ and $t_1(k) = 1$ for each $k \geq 1$.
Then,
the rule $(t_1,0)$-NAV,
where we write $0$ to mean the function that takes value $0$ for all its
inputs, 
seeks committees where each voter approves at least
one committee member. In consequence, the committee that consists of
all candidates is always winning under this rule (and, of course,
also polynomial-time computable). However, it is far more
interesting to seek the smallest $(t_1,0)$-NAV winning committee and
we refer to the rule that outputs such committees as the Minimum
Representation Rule (MRC). A more intuitive description of
this rule follows.

\begin{description}
\item[Minimal Representing Committee rule (MRC).] Under the MRC rule,
  we output all the committees of smallest size such that each voter
  (with a nonempty approval ballot) approves at least one of the committee members.
\end{description}

Intuitively, MRC is very close to the approval variant of the
Chamberlin--Courant
rule~\cite{cha-cou:j:cc,pro-ros-zoh:j:proportional-representation,bet-sli-uhl:j:mon-cc};
we refer to the approval-based Chamberlin--Courant rule as CC. Under
CC, we are given an approval election $E = (C,V)$, a
committee size $k$, and our goal is to find a committee of size $k$
such that as many voters as possible approve at least one of the
committee members (for the case of CC, typically the fact
that a voter approves a candidate is interpreted as saying that the
voter would feel represented by this candidate).  MRC is, in a sense,
a variant of CC where we insist that each voter be represented, but we
want to keep the committee as small as possible.

Since computing an MRC winning committee means, in essence, solving
the minimization version of the \textsc{Set Cover} problem, we next proposition follows
(missing hardness proofs are available in the supplementary material).
\begin{proposition}
  Given an election $E$ and a positive integer $k$, it is $\np$-hard
  to decide if there is an MRC winning committee of size at most $k$.
\end{proposition}
\begin{proof}
  It suffices to note that our problem is equivalent to the \textsc{Set Cover} problem.
  To see this,
  consider a \textsc{Set Cover} instance with $I =
  (U,\calS, k)$, where $U = \{u_1, \ldots, u_n\}$ is a set of
  elements, $\calS = \{S_1, \ldots, S_m\}$ is a family of subsets of
  $U$, and $k$ is an integer. We form an election $E = (C,V)$ where
  for each set $S_i$ we have a candidate $s_i$ and for each element
  $u_j$ we have a voter that approves exactly those candidates $s_i$
  for which $u_j \in S_i$. There is a winning MRC committee of size at
  most $k$ if and only if there is a collection of at most $k$ sets
  that cover $U$.
\end{proof}

Fortunately, computing MRC winning committees is fixed-parameter
tractable (is in $\fpt$) when parameterized by either the number of
candidates or the number of voters (we omit the proof due to space
restriction, but mention that the ideas are similar to those that we use for
Theorem~\ref{thm:threshold-ilp}).

\begin{proposition}\label{mrchard}
  The problem of deciding if there is an MRC winning committee of size
  at most $k$ (in a given election $E$) is in $\fpt$, when
  parameterized either by the number of candidates or the number of
  voters.
\end{proposition}
\begin{proof}
  The result for the number of candidates follows via a straightforward
  brute-force algorithm.
  For parameterization by the number of voters, we invoke the
  ``candidate types'' idea of Chen et
  al.~\shortcite{che-fal-nie-tal:c:few-voters}: There are at most
  $2^n$ ``candidate types'' (where the type of a candidate is simply
  the set of voters that approve of him or her).
  Then, we observe that it suffices to consider at most one candidate of each type,
  since a winning committee certainly never contains two candidates of the same type because we
  could remove one). In $\fpt$-time, we try all possible committees of
  at most $2^n$ candidates (of different types).
\end{proof}

\begin{experiment}
  By applying Experiment~\ref{exp:main} to MRC, we obtain that the
  average committee size is $2.68$ with standard deviation $0.46$.
  Since the rule is $\np$-hard, we have used the brute-force algorithm
  to try all possible committees. We also present results for other
  probabilities of approving each candidate (see
  Figure~\ref{fig:size}). A positive feature of this rule is that the
  size of a winning committee does not depend much on the number of
  voters.
\end{experiment}

We can also use the standard greedy algorithm for \textsc{Set Cover}
to find approximate MRC committees; indeed, we view this algorithm
as a voting rule in its own right.

\begin{description}
\item[GreedyMRC.] Under GreedyMRC, we start with an empty committee
  and perform a sequence of iterations. In each iteration we (a) add
  to the current committee a candidate $c$ that is approved by the
  largest number of voters, and (b) we remove the voters that approve
  $c$ from consideration. After we have removed all voters with
  nonempty approval ballots, we output the resulting committee
  (formally, the rule outputs all the committees that can be obtained
  by breaking the internal ties in some way).
\end{description}

\begin{experiment} By connection to \textsc{Set Cover}, GreedyMRC is
  guaranteed to find a committee that is at most a factor $O(\log m)$
  larger than the one given by the exact MRC (where $m$ is the number
  of candidates).  In our experiment, with approval probabilities in
  $\{0.05, 0.1, \ldots, 0.95\}$, %
  the average sizes of the GreedyMRC committees where no more than 8\%
  larger than the average sizes of the MRC ones (for the case of 20
  voters) or no more than 11\% larger (for the case of 100 voters).
\end{experiment}

MRC and GreedyMRC appear to
be  well suited for choosing small committees of representative;
our experiments confirm this intuition.\medskip

Recall that the function $t_1(x)$ is such that 
$t_1(0)=0$ and $t_1(k) = 1$ for each $k \geq 1$.
Then, consider the $(0,t_1)$-NAV rule, which elects all the
committees that contain candidates approved by all the voters.  While
the empty set is trivially a winning committee under this rule, it is
more interesting to ask about the largest winning committee; we refer
to the rule that outputs the largest winning $(0,t_1)$-NAV rule as the
unanimity rule:

\begin{description}
\item[Unanimity Voting (UV).] Under the unanimity rule, we output the
  committee of all the candidates approved by all the voters.
\end{description}

While computing the smallest $(t_1,0)$-NAV
winning committee is hard (Proposition~\ref{mrchard}), it is easy to compute the (unique) largest
$(0,t_1)$-NAV winning committee in polynomial time (i.e., there is
a polynomial-time algorithm for UV).

\begin{experiment}
  As expected, in our experiment it never happened that some candidate
  was approved by all the voters (the probability of some candidate
  being approved by all $20$ voters is $20\cdot 2^{-20}$ %
  and we considered only 10,000 elections).
\end{experiment}

Both for $(t_1,0)$-NAV and for $(0,t_1)$-NAV, it is trivial to compute
\emph{some} winning committee (the set of all candidates in the former
case and the empty set in the latter). In general, however, this is not 
the case.

\begin{theorem}\label{thm:hard-nav}
  There exists an $(f,g)$-NAV rule for which deciding if there exists
  a committee with at least a given score is $\np$-hard.
\end{theorem}

\begin{proof}
  We consider specific functions $f$ and $g$ and show that for the
  corresponding $(f,g)$-NAV rule it is $\np$-hard to decide if there
  exists a committee with at least a given score.  The specific
  functions $f$ and $g$ we consider are as follows:
\[
  f(x) =
  \begin{cases}
    0, & x = 0 \\
    4, & x \geq 1
  \end{cases}
\quad\quad  g(x) =
  \begin{cases}
    0, & x = 0 \\
    1, & x = 1 \\
    2, & x \geq 2
  \end{cases}
\]

To show $\np$-hardness,
we reduce from the $\np$-hard \textsc{X3C} problem.
In it, we are given sets $\mathcal{S} = \{S_1, \ldots, S_n\}$ over elements $b_1, \ldots, b_n$.
Each set contains exactly three integers and each elements is contained in exactly three sets.
The task is to decide whether there is a set of sets $S' \subseteq \mathcal{S}$ such that each element $b_i$
is covered exactly once.
We assume, without loss of generality, that $n > 39$.

Given an instance of \textsc{X3C} we create an election as follows.
For each set $S_j$ we create a candidate $S_j$.
For each element $b_i$ we create three voters: $v_i^1$ $v_i^2$, and $v_i^3$;
$v_i^1$ and $v_i^2$ are referred to as \emph{set voters} while $v_i^3$ is referred to as \emph{antiset voter}.
Both voters $v_i^1$ and $v_i^2$ approve exactly the candidates corresponding to the sets which contain $b_i$,
while the voter $v_i^3$ approves exactly the candidates corresponding to the sets which do not contain $b_i$.
With respect to the reduced election,
we ask whether a committee with score at least $7n$ exists.
This finishes the description of the reduction.
Next we prove its correctness.

Let $C$ be a committee for the reduced election and let $b_i$ be some element of the \textsc{X3C} instance.
First we show that $C$ has at least six candidates in it.
If it is not the case,
then, since each set $S_j$ covers exactly three elements,
it follows that there are at least $n - 18$ elements not covered by $C$.

Let $b_i$ be an element not covered by $C$. Then, the voters $v_i^1$, $v_i^2$, and $v_i^3$ corresponding to $b_i$
give at most $2$ points to $C$. If $C = \emptyset$ then each voter corresponding to $b_i$ gives $0$ points to $C$.
Otherwise, if $C \neq \emptyset$, then each set voter (each of $v_i^1$ or $v_i^2$) gives at most $-1$ points to $C$,
while the antiset voter gives at most $4$ points. Thus, the voters corresponding to $b_i$ give at most $2$ points to $C$.

There are at most $18$ elements which are covered by $C$.
For each $b_i$ which is covered by $C$, each of the voters $v_i^1$, $v_i^2$, and $v_i^3$ corresponding to $b_i$
give at most $4$ points to $b_i$ (since this is the maximum number of points any voter gives to any committee).
Thus, the voters corresponding to $b_i$ give at most $12$ points to $C$.

Summarizing the above two paragraphs,
we have that the total score of $C$ which has at most six candidates is at most $(3k - 18) \cdot 2 + 18 \cdot 12$.
Since we assume, without loss of generality, that $n  > 39$, we have that this quantity is strictly less than $7n$.
Therefore, from now on we assume that $C$ has at least six candidates in it.

Thus,
let $C$ be a committee with at least six candidates in it and let $b_i$ be an element.
Let $V_i = \{v_i^1, v_i^2, v_i^3\}$
and
consider the following four cases depending on the number of times $b_i$ is covered by the sets $S_j$ corresponding to the candidates in $C$.

\begin{itemize}

\item \textbf{$\boldsymbol b_i$ is not covered by $C$:}
In this case, the score given to $C$ by $V_i$ is at most $(-2) + (-2) + 4 = 0$.
To see this, observe that each of the set voters ($v_i^1, v_i^2$) gives to $C$ exactly $-2$ points,
since they do not approve any candidate from $C$ but disapprove all candidates in $C$;
further, observe that the antiset voter ($v_i^3$) gives to $C$ at most $4$ points,
as this is the maximum number of points any voter can give to any committee.

\item \textbf{$\boldsymbol b_i$ is covered exactly once by $C$:}
In this case, the score given to $C$ by $V_i$ is $2 + 2 + 4 - 1 = 7$.
To see this, observe that each of the set voters ($v_i^1, v_i^2$) gives to $C$ exactly $2$ points,
since they approve one candidate from $C$ (the one candidate corresponding to the one set covering $b_i$)
and disapprove all other candidates in $C$;
further, observe that the antiset voter ($v_i^3$) gives to $C$ exactly $3$ points,
since it approves at least one candidate in $C$ and disapprove exactly one candidate in $C$ 
(the one candidate corresponding to the one set covering $b_i$).

\item \textbf{$\boldsymbol b_i$ is covered more than once by $C$:}
In this case, the score given to $C$ by $V_i$ is $2 + 2 + 4 - 2 = 6$.
To see this, observe that each of the set voters ($v_i^1, v_i^2$) gives to $C$ exactly $2$ points,
since they approve more than one candidate from $C$ (the two or three candidates corresponding to the two or three sets covering $b_i$)
and disapprove all other candidates in $C$;
further, observe that the antiset voter ($v_i^3$) gives to $C$ exactly $2$ points,
since it approves at least one candidate in $C$ and disapprove two or three candidates in $C$ 
(the two or three candidates corresponding to the two or three sets covering $b_i$).

\end{itemize}

As there are exactly $n$ elements,
it follows from the case analysis above that a committee $C$ with score at least $7n$ shall correspond to
an exact cover.
\end{proof}

Naturally, one can come up with many other interesting variants of the
generalized Net-Approval voting rules.
We recommend analysis of this class of rules for future research.

\section{(Net-)Capped Satisfaction and FirstMajority}

Kilgour and Marshall~\cite{kil-mar:b:approval-committees}
introduced the following rule in the context of electing committees of
fixed size, and Kilgour~\cite{kil:j:variable-size-committee}
recalled it in the context of elections with a variable number of
winners, suggesting its net version.
\begin{description}
\item[Capped Satisfaction Approval (CSA).]  The Capped Satisfaction
  Approval (CSA) score of a committee $S$ is defined to be $ \sum_{v_i
    \in V}\frac{|S \cap v_i|}{|S|}.$ The committees with the highest
  score tie as co-winners.

\item[Net Capped Satisfaction Approval (NCSA).]  The NCSA rule uses
  the ``net'' variant of CSA score; specifically, the score of a
  committee $S$ is defined to be 
  $
  \sum_{v_i \in V}\frac{|S \cap  v_i|}{|S|} - \frac{|S \cap \overline{v_i}|}{|S|}
  $
  and the committees with the highest score tie as co-winners.
\end{description}

In the definitions above, the idea behind dividing the scores by the
size of the committee is to ensure that committees which are too
large will not be elected.  Unfortunately, for the rules as defined by
Kilgour~\cite{kil:j:variable-size-committee}, this effect is too
strong, leading mostly to committees containing only the candidate(s)
with the highest approval score. We explain why this is the case and
suggest a modification.

Consider an election $E = (C,V)$ with candidate set $C = \{c_1,
\ldots, c_m\}$ and preference profile $V=(\row vn)$. Let
$\approval(c_1), \ldots, \approval(c_m)$ be the approval scores of the
candidates, and, without loss of generality, assume that
$\approval(c_1) \ge \approval(c_2)\ge \ldots\ge \approval(c_m)$.
Note that, if there are no ties regarding the approval scores, then
for each $k$, the highest-scoring CSA committee of size~$k$ is simply
$S_k = \{c_1, \ldots, c_k\}$ and its score is $ \sum_{v_i \in
  V}\frac{|S_k \cap v_i|}{|S_k|}= \frac{1}{k} \sum_{v_i \in V} |S_k
\cap v_i|=\frac{\approval(c_1)+\ldots+\approval(c_k)}{k}.  $ This
value, however, never increases with $k$ and, so, typically CSA
outputs very small committees (which contain only the candidates with
the highest approval score; the same reasoning applies to NCSA).
Thus, we introduce the $q$-CSA and the $q$-NCSA rules, where $q$ is a
real number, $0 \leq q \leq 1$, and (a) the $q$-CSA score of a
committee $S$ in election $E = (C,V)$ is $ \sum_{v_i \in V}\frac{|S
  \cap v_i|}{|S|^q}$, and (b) the $q$-NCSA score of this committee is
$ \sum_{v_i \in V}\frac{|S \cap v_i|}{|S|^q} - \frac{|S \cap
  \overline{v_i}|}{|S|^q}$. We note that for $q=1$ these rules are,
simply, CSA and NCSA, whereas $0$-NCSA is NAV and $0$-CSA is a rule
that outputs the committee that includes all the candidates that
receive any approvals.

  By the reasoning above, for each rational value of $q$, both $q$-CSA
  and $q$-NCSA are polynomial-time computable (using notation from
  previous paragraph, it suffices to consider the committees $S_1, S_2,
  \ldots, S_m$ and output the one with the highest $q$-CSA or $q$-NCSA
  score, respectively).

  \begin{proposition}
    For each rational value of $q$, there is a polynomial-time
    algorithm that, given an election, computes a winning committee for
    $q$-CSA and for $q$-NCSA.
  \end{proposition}
  
\nimrodbug{  
\begin{figure}
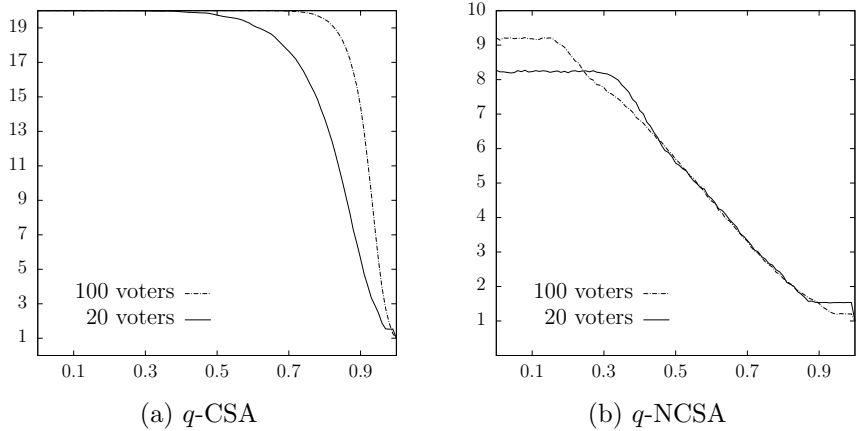

  \center
  \setlength{\tabcolsep}{0pt}
  \begin{tabular}{cc}
  \scalebox{0.6}{\input{csa_size.tex}}&
  \scalebox{0.6}{\input{ncsa_size.tex}}\\
  (a) $q$-CSA &
  (b) $q$-NCSA
  \end{tabular}
  
  \caption{\label{fig:csa-ncsa-sizes}Average committee sizes
    ($y$-axis) under $q$-CSA and $q$-NCSA rules for different values
    of $q$ ($x$-axis); see Experiment~\ref{exp:main} for information
    on how the elections were generated.}
\end{figure}}

\begin{experiment}
  To obtain some better understanding of the influence of the
  parameter $q$ on the size of the committees elected according to $q$-CSA
  and $q$-NCSA, we have repeated Experiment~\ref{exp:main} (with
  approval probability $p = \nicefrac{1}{2}$) for these rules for $q$
  values between $0$ and $1$ with step $0.01$.  The sizes of the
  average committee that we obtained are presented in
  Figure~\ref{fig:csa-ncsa-sizes}. The figures show results for the
  case of $20$ candidates and either $20$ or $100$ voters. While the
  average committee size for $q$-NCSA does not depend very strongly on the
  number of voters (and its dependence on $q$ is appealing), the
  results for $q$-CSA are worrying. Not only does the rule elect
  (nearly) all candidates for most values of $q$, but also for the
  values where it is more selective (e.g., $q=0.9$), the average size
  of its committees depends very strongly on the number of voters.
  In Figure~\ref{fig:size} we show average sizes of
  $0.9$-CSA and of $0.9$-NCSA committees, depending on the probability of
  candidate approval. These figures confirm our worries regarding $q$-CSA
  rules. While the dependence of the average committee size on the
  candidate approval probability for $0.9$-NCSA has the same nature
  irrespective of the number of voters (it is, roughly speaking,
  convex both for $20$ and $100$ voters), the same dependence for
  $0.9$-CSA changes its nature (from roughly convex for the case of
  $20$ voters to roughly concave for the case of $100$ voters).
  In Table~\ref{tab:collection} we also show average committee sizes
  for $0.5$-CSA and $0.5$-NCSA for the candidate approval probability
  $p=\nicefrac{1}{2}$.
\end{experiment}

Given the above experiments, 
we believe that for practical applications, where we may have limited
control on the number of candidates, the number of voters, and the
types of the votes cast, choosing an appropriate value of the
parameter $q$ for $q$-CSA rules (e.g., to promote committees close to
a particular size) would be very difficult. On the other hand,
$q$-NCSA might be robust enough as to be practical.

One could consider generalized variants of the $q$-CSA and $q$-NCSA
rules in the same way as we have considered generalized NAV rules. We
leave this as future work and we conclude the section by considering a
different rule of Kilgour~\cite{kil:j:variable-size-committee},
which is not an (N)CSA rule, but which is somewhat similar since it
also chooses a certain number of candidates with the highest approval
scores.

\begin{description}
\item[FirstMajority.]
  Consider an election with candidates $C = \{c_1, \ldots, c_m\}$.
  For each $c_i \in C$ denote by $s(c_i)$ the approval score of
  $c_i$.  Reorder the candidates so that $\approval(c_1) \geq
  \approval(c_2) \geq \cdots \geq \approval(c_m)$.  The FirstMajority
  rule outputs the smallest committee of the form $\{c_1,c_2, \ldots,
  c_j\}$ such that $\sum_{t=1}^{j} \approval(c_t) > \sum_{q=j+1}^{m}
  \approval(c_q)$.  Note that if some candidates are approved by the
  same number of voters then this rule may return more than one
  committee, corresponding to the various possible ways of reordering the
  candidates.
\end{description}

The very definition of FirstMajority gives a polynomial-time algorithm for
computing its winning committees.

\begin{proposition}
  There is a polynomial-time algorithm that finds some winning
  committee under the FirstMajority rule.
\end{proposition}

\begin{experiment}
  Under our experimental setup (see Experiment~\ref{exp:main}), on the
  average, the FirstMajority rule outputs committees of size $9.51$
  (with standard deviation $0.43$). Further, the size of the committee
  is almost independent of the number of voters and the candidate
  approval probability (see Figure~\ref{fig:size}).
\end{experiment}

\section{Threshold Rules}

We conclude our discussion by considering the threshold rules of
Fishburn and Peke\v{c}~\cite{fishburn2004approval}. 
Let $t \colon \naturals \rightarrow \naturals$ be some function 
referred to as the \emph{threshold function}.  The $t$-Threshold rule
is defined as follows.
\begin{description}
\item[$\boldsymbol{t}$-Threshold.] Consider an election $E =
  (C,V)$. Under the $t$-Threshold rule, we say that a voter $v_i \in
  V$ approves a committee $S$ if $|S \cap v_i| \geq t(|S|)$.  The
  $t$-Threshold rule outputs those committees that are approved by the
  largest number of voters. 
\end{description}
We consider the following three (in some sense extreme) examples of
threshold functions: (a) the unit function $t_\unit = t_1$ (recall the
discussion of generalized net-approval rules);
(b) the majority  function, $t_{\maj}(k) = \nicefrac{k}{2}$;
and (c) the full  function, $t_{\full}(k) = k$.

The $t_\unit$-Threshold rule is very similar to MRC because,
under the $t_\unit$ threshold function, a voter approves a committee
if it includes at least one candidate that this voter approves. Thus
the rule outputs all committees $S$ such that each voter with a
nonempty approval ballot approves some member of $S$ (MRC outputs the
smallest of these committees). Thus finding the largest winning
committee is easy (take all the candidates), but finding the smallest
one is hard (as then we have the MRC rule).

On the other hand, the $t_\full$-Threshold rule 
outputs exactly those committees $S$ that (a) each
candidate in $S$ has the highest approval score and (b) all the
candidates in $S$ are approved by the same group of voters. %
It seems, however, that AV is a simpler and more natural rule than
the $t_\full$-Threshold rule.

Finally, we consider the $t_\maj$-Threshold rule, %
introduced and studied by Fishburn
and Peke\v{c}~\cite{fishburn2004approval}; $t_\maj$-Threshold winning
committees receive broad support from the voters, and---as suggested
by Fishburn and Peke\v{c}---should be ``of moderate size''.  Computing
$t_\maj$-Threshold winning committees is $\np$-hard, but there are
$\fpt$ algorithms.

\begin{theorem}
  The problem of deciding if there is a nonempty committee that
  satisfies all the voters under the $t_{\maj}$-threshold rule is
  $\np$-hard.
\end{theorem}
\begin{proof}
 We give a reduction from the \textsc{Set Cover} problem. Let our
 input instance be $I = (U,\calS,k)$, where $U = \{u_1, \ldots,
 u_n\}$ is a set of elements, $\calS = \{S_1, \ldots, S_m\}$ is a
 family of subsets of $U$, and $k$ is a positive integer. Without
 loss of generality, we can assume that $m > k$ (otherwise there
 would be a trivial solution for our input instance).

 We form an election with the candidate set $C = F \cup \calS$, where $F
 = \{f_1, \ldots, f_k\}$ is a set of filler candidates and $\calS$ is
 a set of candidates corresponding to the sets from the \textsc{Set
   Cover} instance (by a small abuse of notation, we use the same
 symbols for $\calS$ and its contents irrespective if we interpret it
 as part of the \textsc{Set Cover} instance or as candidates in our
 elections).  We introduce $kn+2$ voters:
 \begin{enumerate}
 \item The first voter approves of all the filler candidates and the
   second voter approves all the set candidates. We refer to these
   voters as the balancing voters.
 \item For each element $u_i \in U$, we have a group of $k$ voters,
   so that the $j$th voter in this group ($j \in [k]$) approves of
   all the filler candidates except $f_j$, and also of exactly those set
   candidates that correspond to sets containing $u_i$.
 \end{enumerate}
 We claim that there is a nonempty committee $S$ such that every
 voter %
 approves at least half of the members of $S$ (i.e., every voter is
 satisfied) if and only if $I$ is a \emph{yes}-instance.

 Let us assume that $S$ is a committee that satisfies all the voters.
 We note that $S$ must contain the same number of filler and set
 candidates. If it contained more set candidates than filler
 candidates then the first balancing voter would not be satisfied, and
 if it contained more filler candidates than set candidates, then the
 second balancing voter would not be satisfied. Thus, there is a number
 $k'$ such that $|S| = 2k'$, $k' \leq k$, and $S$ contains exactly
 $k'$ filler and $k'$ set candidates.

 We claim that these $k'$ set candidates correspond to a cover of
 $U$. Consider some arbitrary element $u_i$ and some filler candidate
 $f_j$ such that $f_j$ does not belong to $S$ (since $m > k \geq k'$
such candidates must exist). There is a voter that approves all
 the filler candidates except $f_j$ and all the set candidates that
 contain $u_i$. Thus, the committee contains exactly $k'-1$ filler
 candidates that this voter approves and---to satisfy this
 voter---must contain at least one set candidate that contains
 $u_i$. Since $u_i$ was chosen arbitrarily, we conclude that the set
 candidates from $S$ form a cover of $U$. There are at most $k$ of
 them, so $I$ is a \emph{yes}-instance.

 On the other hand, if there is a family of $k' \leq k$ sets
 that jointly cover %
 $U$, then a committee that
 consists of arbitrarily chosen $k'$ filler candidates and the set
 candidates corresponding to the cover satisfies all the voters.
\end{proof}

\begin{theorem}\label{thm:threshold-ilp}
  Let $t$ be a linear function (i.e., $t(k) = \alpha k$ for some
  $\alpha \in [0,1]$).  There are $\fpt$ algorithms for computing the
  smallest and the largest winning committees under the $t$-Threshold
  rule in $\fpt$ time for parameterizations by the number of candidates
  and by the number of voters.
\end{theorem}

\begin{proof}
  For parameterization by the number of candidates it suffices to try
  all possible committees.
  For parameterization by the number of voters, we combine the
  candidate-type technique of Chen et
  al.~\cite{che-fal-nie-tal:c:few-voters} and an integer linear
  programming (ILP) approach.  The \emph{type} of candidate $c$ is the subset
  of voters that approve $c$.  For an election with $n$ voters, each
  candidate has one of at most $2^n$ types. We describe an algorithm
  for computing a committee approved by at least $N$ voters (where $N$
  is part of the input; it suffices to try all values of $N \in [n]$
  to find a committee with the highest score).  We focus on computing
  the largest winning committee.

  Let $E = (C,V)$ be the input election with $n$ voters.  We form an
  instance of the ILP problem as follows.  For each candidate type
  $i$, $i \in [2^n]$, we introduce integer variable $x_i$ (intuitively
  $x_i$ is the number of candidates of type $i$ that are included in
  the winning committee). For each $i \in [2^n]$, we form constraint
  $0 \leq x_i \leq n_i$, where $n_i$ is the number of candidates of
  type $i$ in election $E$. We also add constraint $\sum_{i \in
    [2^n]}x_i \geq 1$ as the winning committee must be nonempty.

  For each voter $j \in [n]$, we define an integer variable $v_j$
  (the intention is that $v_j$ is $1$ if the $j$th voter approves of
  the committee specified by variables $x_0, \ldots, x_{2^n-1}$ and it
  is $0$ otherwise; see also comments below).  For each $j \in [n]$,
  we introduce constraints $0 \leq v_j \leq 1$, and:
  \begin{equation}
  \label{eq:ilp}
  \textstyle
  \left(\sum_{i \in \textrm{types}(v_j)} x_i\right) -
  t\left(\sum_{i \in [2^n]} x_i\right) \geq -(1-v_j)n,
  \end{equation}
  where $\textrm{types}(v_j)$ is the set of all candidate types
  approved by voter $v_j$. To understand these constraints, note that
  $\sum_{i \in [2^n]} x_i$ is the size of the selected committee,
  $\sum_{i \in \textrm{types}(v_j)} x_i$ is the number of committee
  members approved by the $j$th voter, and, thus, Eq.~\eqref{eq:ilp}
  is satisfied either if $v_j = 0$ or $v_j = 1$ and there is an
  integer $k$ such that the $j$th voter approves at least $t(k)$
  members of the selected size-$k$ committee.  We add constraint $v_1
  + \cdots + v_n \geq N$ (i.e., we require that at least $N$ voters
  are satisfied with the selected committee; this also prevents
  satisfying Eq.~\eqref{eq:ilp} by setting $v_j = 0$ for all $j \in
  [n]$).
  
  To compute the largest committee approved by at least $N$ voters, we
  find a feasible solution (for the above-described integer linear
  program) that maximizes $\sum_{i \in [2^n]}x_i$ (we use
  the famous $\fpt$-time algorithm %
  of Lenstra~\cite{len:j:integer-fixed}).
\end{proof}

\nimrodbug{

\begin{figure}
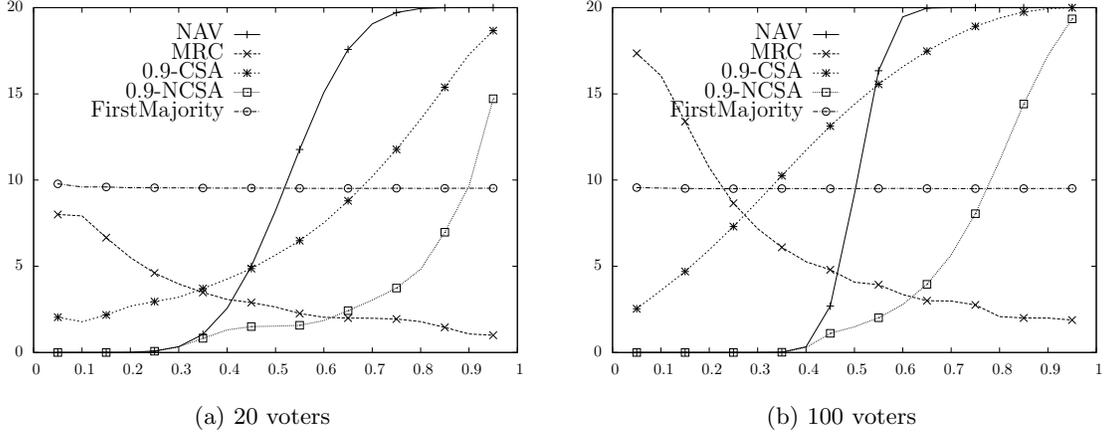

  \begin{subfigure}{0.5\textwidth}
  \scalebox{0.6}{\input{size-20.tex}} 
  \caption{\label{fig:20}20 voters}
  \end{subfigure}
  \begin{subfigure}{0.5\textwidth}
  \scalebox{0.6}{\input{size-100.tex}} 
  \caption{\label{fig:100}100 voters}
  \end{subfigure}
  \caption{\label{fig:size}Average committee sizes for some of our rules (20 candidates and either 20 or 100 voters; approval
    probability is on the $x$-axis).
  }

\end{figure}

}

\begin{table}[t!]
 \center
  \setlength{\tabcolsep}{4pt}
  \begin{tabular}{l|ccc}
    \toprule
    \multirow{2}{*}{rule} & \multicolumn{2}{c}{avg. committee size} & \multirow{2}{*}{complexity}\\
     & \multicolumn{2}{c}{$\pm$ its std. deviation} & \\
    \midrule
          & (20 voters) & (100 voters) & \\[1mm]
    $\nicefrac{2}{3}$-NAV    & $1.02 \pm 1.01$  & $0.01\pm 0.09$& $\p$ \\
    AV   & $1.52 \pm 0.89$   & $1.20 \pm 0.50$& $\p$ \\
    $0.9$-NCSA & $1.52 \pm 0.89$ & $1.50 \pm 0.78$ & $\p$ \\
    MRC  & $2.63 \pm 0.48$   & $4.08 \pm 0.26$ & $\np$-hard \\
    GreedyMRC & $2.75 \pm 0.46$ & $4.55 \pm 0.53$ & $\p$ \\
    $t_{\maj}$-Thr (min) & $2.75 \pm 1.33$ & $\boldsymbol{2.05 \pm 0.34}$ & $\np$-hard \\
    $0.5$-NCSA & $5.57 \pm 2.14$ & $5.57 \pm 2.18$ & $\p$ \\
    $0.9$-CSA & $5.63 \pm 3.02$ & $\boldsymbol{14.67 \pm 2.75}$ & $\p$ \\
    $t_{\maj}$-Thr (max) & $7.68 \pm 3.27$ & $\boldsymbol{2.20 \pm 0.78}$  & $\np$-hard \\
    NAV  & $8.25 \pm 2.19$  & $9.19 \pm 2.23$ & $\p$ \\
    FirstMajority & $9.51 \pm 0.43$ & $9.50 \pm 0.25$ & $\p$ \\
    $0.5$-CSA & $19.74 \pm 0.52$ &  $20.00 \pm 0.00$ & $\p$ \\
    \bottomrule
  \end{tabular}
  \caption{\label{tab:collection}Average committee sizes (see Experiment~\ref{exp:main} for information
    on how the elections were generated). Rules are sorted with respect to the average committee size for $20$ voters (results in bold
    are those that would change their position if we sorted for the average committee size with $100$ voters).
    $t_\maj$-Thr. (min) and $t_\maj$-Thr (max) refer to the smallest and largest committees under the 
    $t_\maj$-Threshold rule.}
\end{table}

\begin{experiment}
  In our experiments, the average size of the smallest
  $t_\maj$-Threshold committee was $2.84$. On the other hand, the
  largest committee contained, on average, $7.52$ candidates. Yet, for
  the case of $100$ voters the difference between the sizes of the
  largest committee and the smallest committee are much more modest
  (see Table~\ref{tab:collection}).
\end{experiment}

Largest winning committees under $t_\maj$-Threshold
are typically of even size (if $S$ is a winning committee of odd size
then it still wins after adding an arbitrary candidate).

\section{Conclusion and Further Research}

We have argued that elections with variable number of winners are
useful and we have analyzed a number of such rules already present in
the
literature %
and provided generalizations for some of them, finding polynomial
algorithms in most cases, but also identifying interesting $\np$-hard
rules.
Further analysis (both axiomatic and computational)
is the most pressing direction for future research.
We also note that in many practical cases there is a societal
preference on the size of the committee to be elected, which is usually
single-peaked. Incorporating this preference into the voting rules is
an interesting direction of research.  
Finally,
as in principle multiwinner voting rules with variable number of winners
cluster the candidates into two sets (those in the winning committee, and the rest),
adapting ideas from data and cluster analysis might prove useful in designing other rules better tailored for this task.

\bibliography{grypiotr2006}

\end{document}